\documentclass[11pt]{article}
\usepackage{fullpage}
\usepackage{amsfonts}
\usepackage{algorithmic}
\usepackage{algorithm}
\usepackage{amsmath,amsthm,amssymb}
\usepackage{latexsym}
\usepackage{epic}
\usepackage{epsfig}
\usepackage{amscd}
\usepackage{url}
\usepackage{verbatim}

\newtheorem{theorem}{Theorem}[section]

\newtheorem{lemma}[theorem]{Lemma}

\newtheorem{definition}[theorem]{Definition}

\newtheorem{fact}[theorem]{Fact}

\newenvironment{proofof}[1]{\begin{proof}[Proof of #1]}{\end{proof}}
\newenvironment{proofsketch}{\begin{proof}[Proof Sketch]}{\end{proof}}

\newtheorem{prop}[theorem]{Proposition}

\newcommand{\set}[1]{\left\{ #1 \right\}}
\newcommand{\union}{\cup}
\newcommand{\intersect}{\cap}
\newcommand{\sm}{\setminus}

\def\min{\qopname\relax n{min}}
\def\max{\qopname\relax n{max}}
\def\argmin{\qopname\relax n{argmin}}
\def\argmax{\qopname\relax n{argmax}}

\def\Pr{\qopname\relax n{\mathbf{Pr}}}
\def\Ex{\qopname\relax n{\mathbf{E}}}

\newcommand{\RR}{\mathbb{R}}

\newcommand{\ZZ}{\mathbb{Z}}
\newcommand{\QQ}{\mathbb{Q}}

\def\A{\mathcal{A}}

\def\L{\mathcal{L}}

\def\sse{\subseteq}

\newcommand{\grad}{\bigtriangledown}

\renewcommand{\comment}[1]{}

\title{Submodular Functions: Extensions, Distributions, and
  Algorithms\\ \vspace{0.1in} A Survey\footnote{This is a slightly revised version of the author's PhD Qualifying Exam Report, the original version of which was submitted to the Department of Computer Science at Stanford University in December of 2009. The Qualifying exam committee consisted of Serge Plotkin, Tim Roughgarden (Thesis Advisor) and Jan Vondr{\'a}k.}} \author{ Shaddin Dughmi \vspace{0.2in}\\
}

\begin{document}

\maketitle

\section{Introduction}
Submodularity is a fundamental phenomenon in combinatorial
optimization. Submodular functions occur in a variety of combinatorial
settings such as coverage problems, cut problems, welfare
maximization, and many more. Therefore, a lot of work has been
concerned with maximizing or minimizing a submodular function, often
subject to combinatorial constraints. Many of these algorithmic
results exhibit a common structure. Namely, the function is
\emph{extended} to a continuous, usually non-linear, function on a
convex domain. Then, this relaxation is solved, and the fractional
solution rounded to yield an integral solution. Often, the continuous
\emph{extension} has a natural interpretation in terms of
distributions on the ground set. This interpretation is often crucial
to the results and their analysis. The purpose of this survey is to
highlight this connection between extensions, distributions,
relaxations, and optimization in the context of submodular functions.

\paragraph{Contributions} 
The purpose of this survey is to present a common framework for viewing
many of the results on optimizing submodular functions. Therefore,
most of the results mentioned -- with the exception of those in
Section \ref{sec:minimizing_cardinality} -- are either already
published, folklore, or easily gotten by existing techniques. In the
first case, citations are provided. Nevertheless, for most of these
results we present alternate, hopefully simplified statements and
proofs that present a more unified picture. In Section
\ref{sec:minimizing_cardinality}, we present a new result for
minimizing symmetric submodular functions subject to a cardinality
constraint.

\section{Preliminaries}
\subsection{Submodular Functions}

We begin with some definitions. We consider a \emph{ground set} $X$ with
$|X|=n$. A \emph{set function} on $X$ is a function $f:2^X \to
\RR$. 

\begin{definition}\label{def:submod1}
A set function $f:2^X \to \RR$ is \emph{submodular} if, for all $A,B
\sse X$
\begin{align*}
  f(A \intersect B) + f(A \union B) \leq f(A) + f(B) 
\end{align*}
\end{definition}

Equivalently, a submodular function can be defined as set function
exhibiting \emph{diminishing margial returns}.

\begin{definition}
  \label{def:submod2}
A set function $f:2^X \to \RR$ is \emph{submodular} if, for all $A,B
\sse X$ with $A \sse B$, and for all $j \in X \sm B$, 
\begin{align*}
  f(A \union \set{j}) - f(A) \geq f(B \union \set{j}) - f(B)
\end{align*}
\end{definition}

The fact that the first definition implies the second can be easily
checked by a simple algebraic manipulation. The other direction can be
shown by a simple induction on $|A \union B| - |A \intersect B|$.

We distinguish additional properties of set functions that will prove
useful. We say $f:2^X \to \RR$ is \emph{nonnegative} if $f(S) \geq 0$
for all $S \sse X$. $f$ is \emph{normalized} if $f(\emptyset) =0$. $f$
is \emph{monotone} if $f(S) \leq f(T)$ whenever $S \sse T$.  Moreover,
$f$ is \emph{symmetric} if $f(S) = f(X \sm S)$ for all $S \sse X$.

Algorithmic results on optimizing submodular functions can often be
stated in the general \emph{value oracle} model. This encapsulates
most special cases of these functions that arise in practice. In the
value oracle model, access to $f$ is via \emph{value queries}: the
algorithm may query for the value of $f(S)$ for any $S$.

We conclude with some more concrete examples of submodular functions
that arise in practice. We say $f$ is a \emph{coverage function} when
elements of $X$ are sets over some other ground set $Y$, and $f(S) =|
\union_{ U \in S} U|$. Therefore, problems such as max-k-cover problem
can be thought of as maximizing a coverage function subject to a
cardinality constraint of $k$. Another class of submodular functions
is cut functions. A set function $f:2^V \to \ZZ$ is a \emph{cut
  function} of a graph $G=(V,E)$ if $f(U)$ is the number of edges of $G$
crossing the cut $(U, V \sm U)$. This can be generalized to
hypergraphs. Moreover, weighted versions of both coverage functions
and cut functions are also submodular. There are many other examples
of submodular functions, for which we refer the reader to the thorough
treatment in \cite{lovasz_submodularity}.

\subsection{Polytopes and Integrality Gaps }

A set $P \sse \RR^n$ is a \emph{polytope} if it is the convex hull of
a finite number of points, known as the \emph{vertices} of $P$, in
$\RR^n$. Equivalently, $P \sse \RR^n$ is a polytope if an only if it
is the intersection of a finite number of halfspaces in
$\RR^n$. Polytopes are convex sets, and are central objects in
combinatorial optimization.

We will consider optimizing continuous functions over polytopes. In the
context of maximization problems, we say a function $F : D \to \RR$
with $P \sse D \sse \RR^n$ has \emph{integrality gap}  $\alpha$
relative to polytope $P$ if \[ \frac{\max\set{F(x) : x \in P}}{\max\set{F(x):
    x \in P, x \in \ZZ^n}} = \alpha \]
If $F$ has integrality gap $1$ relative to $P$, we say it has no
integrality gap relative to $P$.

\subsection{Matroids}
\label{sec:matroids}
A \emph{Set System} is a pair $(X,I)$, where $X$ is the \emph{ground
  set}, and $I$ is a family of subsets of $X$. A special class of set
systems, known as \emph{matroids}, are of particular interest. When
$M=(X,I)$ is a matroid, we refer to elements of $I$ as the
\emph{independent sets} of $M$.
\begin{definition}
  A \emph{matroid} is a set system $(X,I)$ that satisfies
  \begin{itemize}
  \item Downwards Closure: If $T \in I$ and $S \sse T$ then $S \in I$.
  \item Exchange Property: If $S,T \in I$, and $|T| > |S|$, then there
    is $y \in T \sm S$ such that $S \union \set{y} \in I$.
  \end{itemize}
\end{definition}

Given a matroid $M=(X,I)$, we define the \emph{matroid polytope}
$P(M) \sse [0,1]^X$ as the convex hull of the indicator vectors of the
independent sets of $M$.
\[P(M) = hull\left( \set{\vec{\bf 1}_S: S \in I}\right)\]

Edmonds \cite{edmonds_matroids} showed that an equivalent characterization of $P(M)$ can
be given in terms of the \emph{rank function} of the matroid. The rank
function $r_M: 2^X \to \ZZ$ of matroid $M$ is the integer-valued submodular
function defined by \[ r_M(S) = \max\set{|T| : T \in I, T \sse S} \]
Using the rank function, the matroid polytope can be equivalently
characterized as follows. For a vector $x \in \RR^X$ and $S \sse X$,
we use $x(S)$ to denote $\sum_{i \in S} x_i$.

\[ P(M) = \set{x \in \RR+^n : x(S) \leq r_M(S) \mbox{ for all $S \sse X$} } \]
We note that the vertices of the matroid polytopes are all
integers, by the first definition.

\section{Extensions and Distributions}
\label{sec:extensions}
An \emph{extension} of a set function $f: 2^X \to \RR$ is some
function from the hypercube $[0,1]^X$ to $\RR$ that agrees with $f$ on
the vertices of the hypercube. We survey various extensions of
submodular functions, and connect them to distributions on subsets of
the ground set.

\subsection{The Convex Closure, Lov\'{a}sz Extension, and Chain
  Distributions}
\label{sec:convc_lovasz}
In this section, we will define the convex closure of any set
function, and reduce minimization of the set function to minimization
of its convex closure. Then, we will show that, for submodular
functions, the convex closure has a simple form that can be evaluated
efficiently at any point, and thus minimized efficiently.

\subsubsection{The Convex Closure}
\label{sec:convc}
For any set function $f: 2^X \to \RR$, be it submodular or not, we can
define its \emph{convex closure} $f^-$. Intuitively, $f^-$ can be constructed from $f$ by first plotting $f$ in $\RR^{|X| +1}$, and then placing a ``blanket'' under the resulting graph and pulling up until the blanket is taut. Formally, $f^-$ can be defined as follows.
\begin{definition}\label{def:convc1}
  For a set function $f: 2^X \to \RR$, the \emph{convex closure} $f^-:
  [0,1]^X \to \RR$ is the point-wise highest convex function from
  $[0,1]^X$ to $\RR$ that always lowerbounds $f$. 
\end{definition}

It remains to show that convex closure exists and is
well-defined. Observe that the maximum of any number (even infinite)
of convex functions is again a convex function. Moreover, the maximum
of any number (even infinite) of functions lowerbounding $f$ is also a
function lowerbounding $f$. This establishes existence and uniqueness
of the convex closure, as needed. We can equivalently define the convex closure in
terms of distributions on subsets of $X$.

\begin{definition}\label{def:convc2}
  Fix a set function $f: 2^X \to \RR$. For every $x \in [0,1]^X$, let
  $D^-_f(x)$ denote a distribution over $2^X$, with marginals $x$,
  minimizing $\Ex_{S \sim D^-_f(x)}[ f(S) ]$  (breaking ties arbitrarily).
  The \emph{Convex Closure} $f^-$ can be defined as follows: $f^-(x)$
  is the expected value of $f(S)$ over draws $S$ from $D^-_f(x)$.
\end{definition}

To see that the two definitions are equivalent, let us use $f_1^-$ and
$f_2^-$ to denote the convex closure as defined in Definitions
\ref{def:convc1}  and \ref{def:convc2} respectively. First, since the
epigraph of a convex function is a convex set, it
immediately follows that $f_1^-$ lowerbounds $f_2^-$. Moreover, it is
easy to see that $f_2^-(x)$ is the minimum of a simple linear program
with $x$ in the constraint vector; thus $f_2^-$ is convex by
elementary convex analysis. Combining these two facts, we get that
$f^-_1=f^-_2$, as needed.

Next, we mention some simple facts about the convex closure of
$f$. First, it is apparent from Definition \ref{def:convc2} that the
convex closure is indeed an extension. Namely, it agrees with
$f$ on all the integer points. Moreover, it also follows from
Definition \ref{def:convc2} that $f^-$ only takes on values that
correspond to distributions on $2^X$, and thus the minimum of $f^-$ is
attained at an integer point. This gives the following useful
connection between the discrete function and its extension.

\begin{prop}\label{discrete_to_convex}
  The minimum values of $f$ and $f^-$ are equal. If $S$ is a
  minimizer of $f(S)$, then $\vec{\bf 1}_S$ is a minimizer
  of $f^-$. Moreover, if $x$ is a minimizer of $f^-$, then every set in
  the support of $D^-_f(x)$ is a minimizer of $f$.
\end{prop}

\subsubsection{The Lov\'{a}sz Extension and Chain Distributions}
\label{sec:lovasz_ext}
In this section, we will describe an extension $\L_f: [0,1]^X \to
\RR$, defined by Lov\'{a}sz in \cite{lovasz_submodularity}, of an
arbitrary set function $f: 2^X \to \RR$. In the next section we will
show that, when $f$ is submodular, $\L_f= f^-$. We define $\L_f$ as
follows.

\begin{definition}(\cite{lovasz_submodularity})\label{def:lovasz_ext}
  Fix $x \in [0,1]^X$, and let $X=\set{v_1,v_2,\ldots,v_n}$ such that
  $x(v_1) \geq x(v_2) \geq \ldots \geq x(v_n)$. For $0 \leq i \leq n$,
  let $S_i=\set{v_1,\ldots,v_i}$. Let $\set{\lambda_i}_{i=0}^n$ be the
  unique coefficients with $\lambda_i \geq 0$ and $\sum_i \lambda_i
  =1$ such that: \[x= \sum_{i=0}^n \lambda_i 1_{S_i}\] It is easy to
  see that $\lambda_n= x(v_n)$, and for $0\leq i<n$ we have $\lambda_i =
  x(v_i) - x(v_{i+1})$, and $\lambda_0 = 1- x(v_1)$.  The value of the
  \emph{Lov\'{a}sz extension} of $f$ at $x$ is defined as \[ \L_f(x) = \sum_i
  \lambda_i f(S_i) \]
\end{definition}

We can interpret the Lov\'{a}sz Extension as follows. Given a set of
marginal probabilities $x \in [0,1]^X$ on elements of $X$, we
construct a particular distribution $D^\L(x)$ on $2^X$ satisfying
these marginals. Intuitively, this distribution puts as much
probability mass on the large subsets of $X$, subject to obeying the
marginals. Therefore, the largest possible set $X=S_n$ gets as much
probability mass as possible subject to the smallest marginal
$x(v_n)$. When the marginal probability $x(v_n)$ of $v_n$ has been
``saturated'', we put as much mass as possible on the next largest set
$S_{n-1}$. It is easy to see that the next element saturated is
$v_{n-1}$, after we place $x(v_{n-1}) - x(v_n)$ probability mass on
$S_{n-1}$. And so on and so forth. Now, it is easy to see that
$\L_f(x)$ is simply the expected value of $f$ on draws from the
distribution $D^\L(x)$.

A note on the distributions $D^\L(*)$ defining the Lov\'{a}sz
extension. Notice, that the definition $D^\L(x)$ is \emph{oblivious},
in that it does not depend on the particular function $f$. Moreover,
notice that the support of $D^\L(x)$ is a \emph{chain}: a nested
family of sets. We call such a distribution a \emph{chain
  distribution} on $2^X$. The following easy fact will be useful
later.

\begin{fact}\label{fact:chain_is_unique}
  The distribution $D^\L(x)$ is the unique chain distribution on $2^X$
  with marginals $x$.
\end{fact}

\subsubsection{Equivalence of Lov\'{a}sz Extension and Convex Closure}
We will now show that, for a submodular function $f$, the Lov\'{a}sz extension
and the convex closure are one and the same. This is good news, since
we can evaluate the Lov\'{a}sz Extension efficiently at any $x \in
[0,1]^X$, and moreover we can explicitly construct a distribution with
marginals $x$ attaining the value of the Lov\'{a}sz Extension at $x$. This
has implications for minimization of submodular functions, as we will
show in Section \ref{sec:minimizing}.

The intuition behind this equivalence is quite simple.  Recall that,
from Definition \ref{def:convc2}, the value $f^-(x)$ is simply the
minimum possible expected value of $f$ over a distribution on $2^X$
with marginals $x$.  Fixing $f$ and $x$, we ask the question: what
could a distribution $D^-_f(x)$ attaining this minimum look like?
Submodularity of $f$ implies that $f$ exhibits diminishing marginal
returns. Therefore, subject to the marginals $x$, the value of $f$ is
smallest for distributions that ``pack'' as many elements together as
possible in expectation. By definition, that is roughly what $D^\L(x)$
is doing: it packs as many elements together subject to the smallest
marginal, then packs as many unsaturated elements together until the
next marginal is saturated, etc.

While the above intuition is helpful, the proof is made precise by
cleaner \emph{uncrossing} arguments. To illustrate a simple uncrossing
argument, consider two sets $A,B \in 2^X$ that are \emph{crossing}:
neither $A \sse B$ nor $B \sse A$. Now, consider a simple distribution
$D$ that outputs each of $A$ and $B$ with probability $1/2$. Now,
consider \emph{uncrossing} $D$ to form the distribution $D'$, which
outputs each of $A \intersect B$ and $A \union B$ with probability
$1/2$. Observe that $D$ and $D'$ have the same marginals, yet by
direct application of Definition \ref{def:submod1} we conclude that
\[\Ex_{S \sim D'} f(S) = \frac{1}{2} \left( f(A \intersect B) + f(A
  \union B) \right)\leq \frac{1}{2}
\left( f(A) + f(B)\right) \leq \Ex_{S \sim D} f(S) \]

Therefore, starting with any distribution, we can keep uncrossing it
without changing the marginals or increasing the expected value of
$f$. To conclude that this process terminates with a chain
distribution, we need a notion of progress.  We make this precise in
the following Lemma and subsequent Theorem.

\begin{lemma}\label{lem:uncrossing}
  Fix a submodular function $f:2^X \to \RR$. Let $D$ be an arbitrary
  distribution on $2^X$ with marginals $x$. If $D$ is not a chain
  distribution, then there exists another distribution $D'$ with
  marginals $x$ and $\Ex_{S \sim D'} f(S) \leq \Ex_{S \sim D} f(S)$, such
  that $\Ex_{S \sim D'} |S|^2 > \Ex_{S \sim D} |S|^2$.
\end{lemma}
In other words, any non-chain distribution $D$ can be uncrossed to
form a distribution $D'$ that is no worse, and is closer to being a
chain distribution. The quantity $\Ex |S|^2$ is simply a potential
function that measures progress towards a chain distribution; other
choices of potential function work equally well.
\begin{proofof}{Lemma \ref{lem:uncrossing}}
  Fix $f$, $D$ and $x$ as in the statement of the Lemma. Assume $D$ is
  not a chain distribution. Therefore, there exist two sets $A,B \sse
  X$ in the support of $D$ (i.e. $\Pr_D[A],\Pr_D[B] >0$) that are
  \emph{crossing}: neither $A \sse B$ nor $B \sse A$. Assume without
  loss of generality that $\Pr_D[B] \geq \Pr_D[A]$.  We define a new
  distribution $D'$ that simply replaces draws of $A$ and $B$ with
  draws of $A \intersect B$, $B$, and $A \union B$, as follows.
  \[ \Pr_{D'}(S) = \Pr_D[S] \mbox{ for $S \notin \set{A,B,A\intersect
      B, A \union B}$}\]
  \[ \Pr_{D'}(A \intersect B) = \Pr_D[A \intersect B] +\Pr_D[A] \]
  \[ \Pr_{D'}(A \union B) = \Pr_D][A \union B] + \Pr_D[A] \]
  \[ \Pr_{D'}(B) = \Pr_D[B] - \Pr_D[A] \]
  \[ \Pr_{D'}(A) = 0 \]
  
  Notice that distribution $D'$ simply pairs up draws of $A$ and $B$
  from $D$, and replaces each such pair with a draw of $A
  \intersect B$ and a draw of $\A \union B$. It is easy to check that
  this does not change the marginals $x$. Moreover, this allows us to
  conclude that the difference in the expected value of $f$
  is given by:
  \[ \Ex_{S \sim D'}f(S)-\Ex_{S \sim D} f(S) = [\Pr_D[A] f(A
  \intersect B) + \Pr_D[A] f(A \union B)] - [\Pr_D[A] f(A) + \Pr_D[A]
  f(B)] \] Directly applying Definition \ref{def:submod1}, we conclude
  that this quantity is at most $0$. As for the change in the
  potential function $\Ex[|S|^2]$, we get
  \begin{align*}
    \Ex_{S \sim D'} |S|^2 - \Ex_{S \sim D} |S|^2 &= \left(\Pr_D[A] \cdot |A \union
    B|^2 + \Pr_D[A] \cdot | A \intersect B|^2 \right) - \left(
    \Pr_D[A] \cdot |A|^2 + \Pr_D[A]
    \cdot |B|^2 \right) \\
    &= \Pr_D[A] \left( |A \union B|^2 + |A \intersect B|^2 - |A|^2 -
    |B|^2 \right) > 0
\end{align*}
Where the last inequality follows from the inclusion-exclusion
equation and the strict convexity of the squaring function.
\end{proofof}

Now that we know we can ``uncross'' any non-chain distribution without
increasing the expectation of $f$ or changing the marginals, we get
the Theorem.

\begin{theorem} \label{thm:lovasz_is_convex}
  Fix a submodular function $f: 2^X \to \RR$. For any $x \in [0,1]^X$,
  we can take $D^-_f(x) = D^\L(x)$ (without loss), and therefore $f^-(x) =
  \L_f(x)$. Thus $f^- = L_f$. 
\end{theorem}
\begin{proof}
  Fix $f$ and $x$. Let $D^*$ be a choice for $D^-_f(x)$ maximizing
  $E_{S \sim D^*} |S|^2$. The maximum is attained by standard
  compactness arguments. We will show that $D^*$ is a chain
  distribution, which by Fact \ref{fact:chain_is_unique} implies that
  $D^* = D^\L(x)$, completing the proof.

  Indeed, if $D^*$ were not a chain distribution, then by Lemma
  \ref{lem:uncrossing}, there exists another choice $D'$ for
  $D^-_f(x)$ such that $\Ex_{S \sim D'} |S|^2 > \Ex_{S \sim D^*}
  |S|^2$. This contradicts the definition of $D^*$.
\end{proof}

The above Theorem implies the following remarkable
observation: The distribution minimizing a submodular function subject
to given marginals can be chosen \emph{obliviously}, since $D^\L(x)$
does not depend on the particular submodular function $f$ being
minimized. As we will see in the next section, the same does not hold
for maximization.

For completeness, we conclude with a strong converse of Theorem
\ref{thm:lovasz_is_convex}.
\begin{theorem}
  Fix a set function $f:2^X \to \RR$. If $\L_f$ is convex then $f$ is submodular.
\end{theorem}
\begin{proof}
  We take a non-submodular $f$, and show that $\L_f$ is non-convex. We
  will show that the Lov\'{a}sz extension makes a suboptimal choice for
  minimization at some $x\in [0,1]^X$: namely, $\L_f(x) > f^-(x)$. By
  Definition \ref{def:convc1}, $f^-(x)$ is the point-wise greatest
  convex extension of $f$. This implies that $L_f$ is non-convex.

  We now exhibit $x$ such that $\L_f(x) > f^-(x)$. By Definition
  \ref{def:submod2}, there exists a set $A \sse X$, and two elements
  $i,j \notin A$, such that
  \[ f(A \union \set{i,j}) - f(A \union \set{i}) > f(A \union \set{j})
  - f(A)\] Define $x \in [0,1]^X$ as follows: $x(k) = 1$ for each
  $k \in A$, and $x(i)=x(j)=1/2$, and $x(k)=0$ otherwise.  Now it is
  intuitively clear that the Lov\'{a}sz Extension makes the \emph{wrong}
  choice for minimization: it will attempt to bundle $i$ and $j$
  together despite  \emph{increasing} marginal returns. Indeed, By
  Definition \ref{def:lovasz_ext}, the Lov\'{a}sz extension at $x$ evaluates to
  \[ \L_f(x) = \frac{1}{2} f(A \union \set{i,j}) + \frac{1}{2} f(A) \]
  Now, consider the distribution $D$, with marginals $x$, defined by
  $\Pr_D[A \union \set{i}] = \Pr_D[A \union \set{j}] =
  \frac{1}{2}$. Since, by Definition \ref{def:convc2}, $f^-(x)$
  lowerbounds the expectation of any distribution with marginals $x$,
  we have that \[ f^-(x) \leq \frac{1}{2} f(A \union \set{i}) +
  \frac{1}{2} f(A \union \set{j})\] We can now combine the three above
  inequalities to establish $\L_f(x) > f^-(x)$, completing the proof.
  \[2 ( L_f(x) - f^-(x)) \geq f(A \union \set{i,j}) + f(A) - f(A \union
  \set{i}) - f(A \union \set{j}) >0\]
\end{proof}

\subsection{The Concave Closure}\label{sec:concc}
The \emph{concave closure} $f^+$ of any set function $f$ can be defined
analogously to the convex closure. The intuition is similar: $f^+$ can be constructed from $f$ by first plotting $f$ in $\RR^{|X| +1}$, and then placing a ``blanket'' above the resulting graph and pulling down until the blanket is taut. We again state the two equivalent formal definitions.

\begin{definition}\label{def:concc1}
  For a set function $f: 2^X \to \RR$, the \emph{concave closure} $f^+:
  [0,1]^X \to \RR$ is the point-wise lowest concave function from
  $[0,1]^X$ to $\RR$ that always upperbounds $f$. 
\end{definition}
\begin{definition}\label{def:concc2}
  Fix a set function $f: 2^X \to \RR$. For every $x \in [0,1]^X$, let
  $D^+_f(x)$ denote a distribution over $2^X$, with marginals $x$,
  maximizing $\Ex_{S \sim D^+_f(x)}[ f(S) ]$  (breaking ties arbitrarily).
  The \emph{Concave Closure} $f^+$ can be defined as follows: $f^+(x)$
  is the expected value of $f(S)$ over draws $S$ from $D^+_f(x)$.
\end{definition}

By a similar argument to that presented in Section \ref{sec:convc},
both definitions are well-defined and equivalent. 

It is tempting to attempt to explicitly characterize the distribution
$D^+_f(x)$ in the same way we characterized $D^-_f(x)$. However, no
such tractable characterization is possible. In fact, it is NP-hard to
even evaluate $f^+(x)$, even when $f$ is a graph cut function.

\begin{theorem}(\cite{vondrak_ipco, vondrak_thesis})
  It is $NP$-hard to evaluate $f^+(x)$ for an arbitrary submodular
  $f:2^X \to \RR$ and $x\in [0,1]^X$. This is true even when $f$ is a
  graph cut function.
\end{theorem}
\begin{proof}
  The proof is by reduction from the NP-hard problem Max-Cut. In the
  max cut problem, we are given an undirected graph $G=(V,E)$, and the
  goal is to find a cut $(S,V\sm S)$ maximizing the number of edges
  crossing the cut. Let $f(S)$ be number of edges crossing the cut
  $(S,V\sm S)$. 

  We reduce finding the maximum non-trivial cut (with $S
  \neq \emptyset,V$) to the following convex optimization problem:
  Maximize $f^+(x)$ subject to $ 1 \leq \vec{\bf 1}\cdot x \leq n-1$. Indeed,
  it is clear that this is a relaxation of the max-cut problem. The
  optimum is attained at an integer point $x^*$, since without loss of
  generality the trivial sets ($\emptyset$ and $V$) will not be in the
  support of any optimum distribution. Therefore, if $f^+(x)$ can be
  evaluated in polynomial time for an arbitrary $x$, then this convex
  optimization problem can be solved efficiently. This completes the reduction.
\end{proof}

Stronger hardness results are possible. In fact, it is shown in
\cite{vondrak_thesis} that, even when $f$ is a monotone coverage
function and $k$ is an integer, the convex optimization problem
$\max\set{ f^+(x) : \vec{\bf 1} \cdot x \leq k}$ is APX-hard. More
generally, it is shown in \cite{vondrak_nonmonotone} that it is hard
to maximize general submodular functions in the value oracle model
(independently of $P \neq NP$) with an approximation factor better
than $1/2$.

In light of these difficulties, there is no hope of finding exact
polynomial time algorithms for maximizing submodular functions in most
interesting settings, using $f^+$ or otherwise. Therefore, we will
consider another extension of submodular functions that will prove
useful in attaining constant factor approximations for maximization
problems.

\subsection{The Multilinear Extension and Independent Distributions}
\label{sec:multilinear}
\subsubsection{Defining the Multilinear Extension}
Ideally, since concavity is intimately tied to maximization, we could
use the concave extension of a submodular function in relaxations of
maximization problems. However, unlike the convex closure, the concave
closure of a submodular function cannot be evaluated
efficiently. Moreover, since $f^+$ is the point-wise lowest concave
extension of $f$, any concave extension will have a non-trivial
integrality gap relative to most interesting polytopes, including even
the hypercube. In other words, any concave extension other than $f^+$
will not correspond to a distribution at every point of the domain
$[0,1]^X$; a property that has served us particularly well in
minimization problems.

In light of these limitations of concave extensions, we relax this
requirement and instead exhibit a simple extension that is
\emph{up-concave}: concave in all directions $\vec{u} \in \RR^n$ with
$u \succeq 0$ (or, equivalently $u \preceq 0$). Moreover, this
extension will correspond to a natural distribution at every point,
and therefore will have no integrality gap on the domain
$[0,1]^X$. Surprisingly, this extension will also have no integrality
gap over any matroid polytope. As we will see in Section
\ref{sec:maximizing}, it turns out that, under some additional
conditions, up-concave functions can be approximately maximized over a
large class of polytopes.

Without further a-do, we define the \emph{multilinear extension} $F$
of a set function $f$. First, we say a function $F:[0,1]^X \to \RR$ is
\emph{multi-linear} if it is linear in each variable $x_i$, when the
other variables $\set{x_j}_{j \neq i}$ are held fixed. It is easy to
see that multilinear functions from $\RR^X$ to $\RR$ form a vector
space. Moreover, a simple induction on dimension shows that a
multi-linear function is uniquely determined by its values on the
vertices of the hypercube. This allows us to define the multilinear
extension.
\begin{definition}\label{def:multilinear1}
  Fix set function $f: 2^X \to \RR$. The \emph{multilinear extension}
  $F:[0,1]^X \to \RR$ of $f$ is the unique multilinear function
  agreeing with $f$ on the vertices of the hypercube.
\end{definition}

As with the Lov\'{a}sz extension, the multilinear extension corresponds to
a natural distribution at each point $x \in [0,1]^X$, and moreover
this distribution has marginals $x$. This distribution becomes
apparent if we express $F$ in terms of a simple basis, with each
element of the basis corresponding to a vertex of the hypercube. For a
set $S \sse X$, we define the multilinear basis function $M_S$ as follows
\[ M_S(x) = \prod_{i \in S} x_i\cdot \prod_{i \neq
  S}(1-x_i)\]
Since a multilinear function is uniquely determined by its values on
the hypercube, it is easy to check that  any multilinear function can be
written as a linear combination of the basis functions $\set{M_S}_{S
  \sse X}$, with $f(S)$ as the coefficient of $M_S$.

\[F(x) = \sum_{S \sse X} f(S) \cdot M_S(x) = \sum_{S \sse X} f(S)
\cdot \prod_{i \in S} x_i\cdot \prod_{i \neq S}(1-x_i) \]

Inspecting the above expression, we notice that $F(x)$ corresponds to
a simple distribution with marginals at $x$. Let $D^i(x)$ be the
distribution on $2^X$ that simply picks each element $v \in X$
independently with probability $x(v)$. It is clear that $Pr_{D^i(x)}
[S] = M_S(x)$. Therefore, it is clear that $F(x)$ is simply the
expected value of $f$ over draws from $D^i(x)$. This gives the
following equivalent definition of the multilinear extension.
\begin{definition}\label{def:multilinear2}
  Fix a set function $f:2^X \to \RR$. For each $x \in [0,1]^x$, let
  $D^i(x)$ be the distribution on $2^X$ that picks each $v \in X$
  independently with probability $x(v)$. The value of the multilinear
  extension $F:[0,1]^X \to \RR$ at $x$ can be defined as the expected
  value of $f$ over draws from $D^i(x)$. 
\begin{align}\label{eq:multilinear}
  F(x) = \Ex_{S \sim D^i(x)} f(S) = \sum_{S \sse X} f(S) \cdot
  \prod_{i \in S} x_i\cdot \prod_{i \neq S}(1-x_i)
\end{align}
\end{definition}

We note that, like the Lov\'{a}sz extension, the multilinear extension has
the property of being \emph{oblivious}: the distribution defining $F$
at $x$ does not depend on the set function $f$. The fact that, yet
again, an oblivious extension lends itself particularly well to
solving optimization problems is a remarkable, and arguably fundamental
phenomenon.

\subsubsection{Useful Properties of the Multilinear Extension}
In this section, we will develop some properties of the Multilinear
extension that will be useful for problems involving maximization of
submodular functions. The maximization problem we consider in Section
\ref{sec:maximizing} is that of maximizing a monotone, submodular function
$f:2^X \to \RR$ over independent sets of a matroid $M=(X,I)$.

First, we show that the multilinear relaxation of a monotone
set function is also monotone.
\begin{prop}(\cite{vondrak_stoc08})
  If $f:2^X \to \RR$ is a monotone set function, then its multilinear
  relaxation $F: [0,1]^X \to \RR$ is monotone. That is, whenever $x
  \preceq y$, we have $F(x) \leq F(y)$. 
\end{prop}
\begin{proof}
  By Definition \ref{def:multilinear2}, it suffices to show that
  distribution $D^i(y)$ draws a \emph{pointwise} larger set than
  $D^i(x)$. We couple draws $S_x \sim D^i(x)$ and $ S_y \sim D^i(y)$
  in the obvious way: for each $v \in X$ we independently draw a
  random variable $R(v)$ from the uniform distribution on $[0,1]$. If
  $0 \leq R(v) \leq x(v)$, then we let $v \in S_x$, otherwise $v
  \notin S_x$. Similarly, $v \in S_y$ if and only if $0 \leq R(v) \leq
  y(v)$. Since we have $x(v) \leq y(v)$ for every $v \in X$, it is
  clear that, under this coupling, $S_x \sse S_y$ pointwise. By
  monotonicity of $f$, this implies that $F(x) \leq F(y)$.
\end{proof}

Next, we will show that the multilinear extension $F$ of a submodular
$f$ is \emph{up-concave}: concave when restricted to any direction $u
\succeq 0$ (equivalently $u \preceq 0$). In other words, it must be
that for any $x \in [0,1]^X$ and $u \succeq 0$, the expression
$F(x + t u)$ is concave as a function of $t \in \RR$ over the domain
of $F$. This is consistent with the diminishing-marginal-returns
interpretation of submodularity and the independent distribution
interpretation of $F$: weakly increasing the marginal probability of
drawing each item can only result in items getting packed together
into larger and larger sets (in a point-wise sense), and hence
yielding diminishing marginal increases in the expected value of
$f$. This intuition can be made precise by carefully coupling draws
from $D^i(x)$ and $D^i(y)$ for some $x \preceq y$, and considering the
marginal increases from transitioning from $D^i(x)$ and $D^i(y)$ to $D^i(x + \delta
u)$ and $D^i(y + \delta u)$ respectively (for some $u \succeq 0$ and
arbitrarility small $\delta >0$). However, we will instead use tools from
linear algebra to get a cleaner proof.

Using elementary linear algebra, up-concavity can be re-stated as a
condition on the hessian matrix $\grad^2 F(x)$ of $F$ at $x$. The
matrix $\grad^2 F (x)$ is the symmetric matrix with rows and columns
indexed by $X$, and the $(i,j)$'th entry corresponding to the second
partial derivative $\frac{\partial^2 F}{\partial x_i \partial x_j}
(x)$. Up-concavity is then the condition that

\[ u^T (\grad^2 F (x)) u \leq 0 \mbox{ for all $x \in [0,1]^X$ and $u
  \succeq 0$} \]

Since $F$ is multilinear, the diagonal entries of $\grad^2 F(x)$ are
always $0$.  Therefore, by considering $\set{\vec{\bf
    1}_{\set{i,j}}}_{i,j \in X}$ as choices for $u$, we conclude that
$F$ is up-concave if and only if all the second partial derivatives
$\frac{\partial^2 F}{\partial x_i \partial x_j} (x)$ are
non-positive. Indeed, this is consistent with submodularity and the
independent-distribution interpretation of $F$: increasing the
probability of including $i$ only results in sets that are point-wise
larger, and therefore these sets would benefit less by inclusion of
$j$ as well.

\begin{prop}(\cite{vondrak_stoc08})
  If $f:2^X \to \RR$ is submodular, then its multi-linear relaxation
  $F:[0,1]^X \to \RR$ is up-concave.
\end{prop}
\begin{proof}
  By the discussion above it suffices to show that, for each $x \in
  [0,1]^X$, we have that $\frac{\partial^2 F}{\partial x_i \partial
    x_j} (x) \leq 0$. Fixing $x$, we take the derivative of $F$ with respect
  to $x_i$ to get:
\[ \frac{\partial F}{\partial x_i} = \frac{\partial}{\partial x_i}
\Ex_{S \sim D^i(x)} f(S)= \Ex_{S \sim D^i(x)} [f(S \union \set{i})
- f(S)] \]

The equality above follows immediately from the independent
distribution interpretation of $F$, by conditioning on all events $j
\in S$ for $j \neq i$ and considering the expectation of $f$ as a
function of the marginal probability of $i$. Using linearity of
expectation and taking the derivative again in the same way with
respect to $j$, we get

\[ \frac{\partial^2 F}{\partial x_i \partial x_j} = \Ex_{S \sim
  D^i(x)} [f(S \union \set{i,j}) -f(S \union i)
- f(S \union j) + f(S)] \]
Using Definition \ref{def:submod1}, we get that this quantity is
non-positive, as needed.
\end{proof}

It is clear that, since the value of $F$ at any point corresponds to
the expectation of $f$ at a distribution on $2^X$, that $F$ has no
integrality gap relative to the hypercube $[0,1]^X$. Since we will
consider constrained maximization problems, it would be useful if this
held for interesting subsets of the hypercube. It turns out that, for
a submodular function $f$, a useful property that we term
\emph{cross-convexity} yields precisely such a guarantee relative to
all matroid polytopes. Cross convexity means that trading off two
elements $i$ and $j$ gives a convex function, or increasing marginal
returns.

\begin{definition}
  We say a function $F: [0,1]^X \to \RR$ is \emph{cross-convex} if,
  for any $i\neq j$, the function $F^x_{i,j}(\epsilon) := F(x
  +\epsilon (e_i -e_j))$ is convex as a function of $\epsilon\in \RR$.
\end{definition}
Cross-convexity is consistent with submodularity and the independent
distribution interpretation of the multilinear relaxation. Consider
independent distribution $D^i(x)$, and the associated expectation of
$f$. It is an easy exercise to see that the probability of
``collision'' of $i$ and $j$ -- that is, the probability that both are
drawn by the indepenent distribution -- is a \emph{concave} function
of $\epsilon$. Since ``collision'' corresponds to diminishing marginal
returns, or a \emph{decrease} in the expected value of $f$, this means
that the expectation of $f$ is \emph{convex} in $\epsilon$. We make
this precise in the proposition below.
\begin{prop}(\cite{vondrak_ipco})
  When $f:2^X \to \RR$ is submodular, its multilinear extension $F$ is cross-convex.
\end{prop}
\begin{proof}
  Fix $x$ and $i \neq j$. We can write $F^x_{i,j}(\epsilon)$ as:
\begin{align*}
   F^x_{i,j}(\epsilon) &= \Ex_{S \sim D^i(x+ \epsilon (e_i - e_j))}
  f(S)
\end{align*}
Consider the random variable $\hat{S}$: the set of elements other than
$i$ and $j$ that are drawn from $D^i(x + \epsilon (e_i - e_j))$. We
have
\begin{align*}
  &= \Ex_{\hat{S}} [  (x_i + \epsilon)(x_j - \epsilon) f(\hat{S} \union \set{i,j})\\
   &+ (x_i + \epsilon)(1- x_j + \epsilon) f(\hat{S} \union \set{i})\\
   &+ (1-x_i - \epsilon)( x_j - \epsilon) f(\hat{S} \union \set{j})\\
   &+ (1-x_i - \epsilon)( 1- x_j + \epsilon) f(\hat{S}) ]
\end{align*}

Observe that the coefficient of $\epsilon^2$ in the above expression
is $f(\hat{S} \union \set{i}) + f( \hat{S} \union \set{j}) - f( \hat{S}
\union \set{i,j}) - f(\hat{S})$. By Definition \ref{def:submod1}, this
is nonnegative, which yields convexity in $\epsilon$ as needed.
\end{proof}

Consider any $x \in [0,1]^X$ and any fractional $x_i,x_j$. We can
trade off items $i$ and $j$, in the sense defined above, until one of
them is integral. Cross convexity implies that the maximum point of
this tradeoff lies at the extremes. Therefore, repeating this process
as long as there are fractional variables, we can arrive at an integer
point $x' \in \set{0,1}^X$ such that $F(x') \geq F(x)$. When the set
of feasible solutions is constrained to a proper subset of the
hypercube, however, this may result in an infeasible
$x'$. Nevertheless, for well-structured matroid polytopes, a careful
rounding process maintains feasibility without decreasing the
objective value. This is known as \emph{Pipage rounding}, and will be
presented in Section \ref{sec:maximizing}.



\section{Algorithmic Implications}
In this section, we consider minimization and maximization problems
for submodular functions. The algorithms we consider will make heavy
use of the extensions described in Section \ref{sec:extensions}. 

The algorithms we consider take as input a set $X$ with $|X|=n$, and a
rational number $B$. For the maximization problem we consider,
additional constraints are given as input; we defer details to Section
\ref{sec:maximizing}. The function $f:2^X \to \QQ$ to be minimized or
maximized is assumed to satisfy $\frac{\max_{S \sse X} f(S)}{\min_{S
    \sse X} f(S)} \leq B$. The algorithms we present will operate in
the value oracle model. We require that the algorithms run in time
polynomial in $n$ and $\log B$, and therefore also make a polynomial
number of queries to the value oracle.

\subsection{Minimizing Submodular Functions}
\label{sec:minimizing}
Proposition \ref{discrete_to_convex} allows us to reduce discrete
optimization to continuous optimization. Namely, we reduce
minimization of $f$ to minimization of its convex closure $f^-$. When
$f^-$ can be evaluated efficiently, this yields an efficient algorithm
for minimizing $f$ using the standard techniques of convex
optimization.

When $f$ is submodular, $f^-= \L_f$. It is clear from Section
\ref{sec:lovasz_ext} that the Lov\'{a}sz extension can be evaluated
efficiently: we can explicitly construct the distribution $D^\L(x)$,
which has support of size at most $n+1$, and then explicitly compute
the expected value of $f$ over draws from $D^\L(x)$. Therefore, we can
compute the minimum of a submodular function by finding the minimum of
its Lov\'{a}sz extension.

\begin{theorem}(\cite{gls,lovasz_submodularity}) There exists an algorithm for
  minimizing a submodular function $f:2^X \to \RR$ in the value query
  model, running in time polynomial in $n$ and $\log B$.
\end{theorem}

\subsection{Maximizing Monotone Submodular Functions Subject to a Matroid Constraint}
\label{sec:maximizing}

In this section, we consider the probelm of maximizing a nonnegative,
monotone, submodular function $f:2^X \to \RR$ over independent sets of
a matroid $M=(X,I)$. We assume $f$ is given by a value oracle as
usual, and $M$ is given by an independence oracle: An oracle that
answers queries of the form: is $S \in I$? It is well known that much
can be accomplished in this independence oracle model. In particular,
we can use submodular function minimization, presented in Section
\ref{sec:minimizing}, to get a separation oracle for the matroid
polytope $P(M)$.

First, we begin where we left off in Section
\ref{sec:multilinear}. Namely, we will show that we can indeed reduce
maximization of $f$ over $M$ to maximization of the multilinear
relaxation $F$ over the polytope $P(M)$. In particular, we show that
$F$ has no integrality gap relative to $P(M)$, and the rounding can be
done in polynomial time. This is known as \emph{Pipage Rounding}.

\begin{lemma}(\cite{vondrak_ipco}) \label{lem:matroid_nogap}
  Fix a submodular function $f: 2^X \to \RR$ and its multilinear
  relaxation $F$. Fix a matroid $M=(X,I)$. For every point $x \in
  P(M)$, there exists an integer point $x' \in P(M)$ such that $F(x')
  \geq F(x)$. Therefore, $F$ has no integrality gap relative to
  $P(M)$. Moreover, starting with $x$, we can construct $x'$ in
  polynomial time.
\end{lemma}
\begin{proof}
  Recall that the rank function $r_M$ of matroid $M$ is an integer
  valued, normalized, monotone, and submodular set function. Moreover,
  recall that the matroid polytope is as defined in
  \ref{sec:matroids}. In the ensuing discussion, we will
  assume that we can efficiently check whether $x \in P(M)$, and moreover we can
  find tight constraint when $x$ is on the boundary of $P(M)$. Both
  problems are solvable by submodular function minimization.

  By multilinearity, the proposition is trivial when there is only a
  single fractional variable. Moreover, by multilinearity we may
  assume without loss of generality that every fractional variable
  appears in at least one tight constraint of the matroid polytope.

  It follows from the submodularity of $r_M$ that the family of
  ``tight sets'', those sets $S \sse X$ with $x(S) = r_M(S)$, is
  closed under intersection and union. Therefore, we consider a
  \emph{minimal} tight set $T$ with fractional variables $x_i$ and
  $x_j$, and trade off $x_i$ and $x_j$ subject to not violating
  feasibility (i.e. not leaving the matroid polytope $P(M)$). Observe
  that, by cross-convexity, we can choose the extreme point of this
  tradeoff so that the value of $F$ does not decrease. Moreover, one
  of two types of progress is made: either an additional variable is
  made integral, or a new tight set $T'$ is created that includes
  exactly one of $i$ or $j$.  It remains to show that, repeating this
  process so long as there are fractional variables, the second type
  of progress can occur consecutively at most $n$ times. This would
  complete the proof, showing that after at most $n^2$ steps all
  variables are integral.

  Observe that, since $T$ was chosen to be minimal and the tight sets
  are closed under intersection, trading off $x_i$ and $x_j$ does not
  ``untighten'' any set. Therefore, this process can only grow the
  family of tight sets. For simplicity, we assume that at each step we
  choose $T$ to be a tight set of minimum cardinality. (This assumption
  can be easily removed by more careful accounting.) If no variable is
  made integral after trading off $x_i$ and $x_j$, then an additional
  tight set $T'$ is created that includes exactly one of $i$ or
  $j$. Since tight sets are closed under intersection, and tight sets
  are preserved, this implies that the cardinality of smallest tight
  set strictly decreases. Therefore, a variable must be made
  integral after at most $n$ iterations, completing the proof.
\end{proof}

Now, it remains to show that $F$ can be maximized approximately over
$P(M)$. In fact, something even more general is true, as shown by
Vondr{\'a}k in \cite{vondrak_stoc08}: Any nonnegative, monotone, up-concave
function can be approximately maximized over any solvable packing
polytope contained in the hypercube. Here, by \emph{packing polytope}
we mean a polytope $P \sse [0,1]^X$ that is \emph{down monotone}: If
$x,y \in [0,1]^X$ with $x \preceq y$ and $y \in P$, then $x\in P$. A
polytope $P$ is \emph{solvable} if we can maximize arbitrary linear
functions over $P$ in polynomial (in $n$) time, or equivalently if $P$
admits a polynomial time separation oracle.

\begin{lemma}(\cite{vondrak_stoc08}) \label{lem:maximize_upconcave}
  Fix a solvable packing polytope $P \sse [0,1]^X$. Fix a nonnegative,
  monotone, up-concave function $F: [0,1]^X \to \RR+$, that can be
  evaluated at an arbitrary point in polynomial time. Then the problem
  $\max\set{F(x) : x \in P}$ can be approximated to within a factor of
  $1-1/e$ in polynomial time.
\end{lemma}
\begin{proofsketch}
  We may assume without loss of generality that $F(\vec{0}) =0$. We
  let $OPT$ denote the maximum value of $F$ in $P$, and use $x^*$ to
  denote the point in $P$ attaining this optimal. Since $F$ is not
  concave in all directions, usual gradient descent techniques fail to
  provide any guarantees. Instead, we will show a modified
  gradient-descent-like technique that exploits up-concavity. We will
  consider a particle with starting position at $\vec{0} \in P$, and
  slowly move the particle in \emph{positive directions} only:
  directions $u \in \RR+^n$. This restriction is not without loss: any
  local descent algorithm that does not ``backtrack'' cannot guarantee
  finding the optimal solution. Nevertheless, by arguments analogous
  to those used for the greedy algorithm for max-k-cover, we can
  guarantee a $1-1/e$ approximation. We assume the motion of the
  particle is a continuous process, ignoring technical details related
  to discretizing this process so that it can be simulated in
  polynomial time.

  We use $x(t)$ to denote the position of the particle at time $t$. We
  interpret the position $x(t)$ of the particle as a convex
  combination of vertices $V_P$ of $P$, with vertex $v \in P$ having
  coefficient $\alpha_v(t)$
  \[ x(t) = \sum_{v \in V_P} \alpha_v(t) \cdot v\] Initially,
  $\alpha_{\vec{0}}(0)=1$, and $\alpha_v(0) = 0$ for each vertex $v
  \neq \vec{0}$. So long as $\alpha_{\vec{0}}(t) >0$, there is room
  for improvement in positive directions: we can replace $\vec{0}$ in
  the convex combination by some other vertex $z \succeq \vec{0}$. By
  monotonicity, this increases the value of $F$.  

  More concretely, for a small $d_t>0$, we let $\alpha_{\vec{0}}(t +
  d_t) = \alpha_{\vec{0}}(t) - d_t$, and $\alpha_z(t + d_t) =
  \alpha_z(t) + d_t$. We keep $\alpha_v (t+d_t) = \alpha_v(t)$ for all
  $v \neq \vec{0},z$. It is clear that this process must terminate
  when $t=1$, since at that point the vertex $\vec{0}$ is no longer
  represented in the convex combination. It remains to show how to
  choose $z$ at each step so that $F(x(1)) \geq (1-1/e) OPT$. By
  simple calculus, it suffices to show that $z$ can be chosen so that
  $\frac{d F(x(t))}{d_t} \geq OPT - F(x(t))$. In other words, that the
  rate of increase in the objective is proportional to the distance
  from the optimal. This is analogous to the analysis of many discrete
  greedy algorithms, such as that for max-k-cover.

  Fixing a time $t$, what if we choose $z$ so as to maximize the local
  gain? In other words, \[z = \argmax_{z \in P} \grad F(x(t)) \cdot
  z\] Finding such a $z$ reduces to maximizing a linear function over
  the matroid polytope, which can be accomplished in polynomial
  time. It remains to show that there exists a $z' \in P$ with $\grad
  F(x(t)) \cdot z' \geq OPT - F(x(t))$.

  Consider $z' = \max(x(t),x^*) - x(t)$, where the maximization is
  taken co-ordinate wise. We can interpret $z'$ as the ``set-wise
  difference'' betweeen $x^*$ and $x(t)$. Indeed, if $x^*$ and
  $x(t)$ were integral indicator vectors corresponding to subsets of
  $X$, then $z'$ is precisely the indicator vector of their set
  difference. The difference between any two sets in a
  downwards-closed set system is again in the set system. This analogy
  can be made precise to show that $z' \in P$ as follows: $z' \preceq
  x^* \in P$.

  We now show that $z'$ gives the desired marginal increase in
  objective. First, it is easy to see that $x(t) + z' \succeq x^*$,
  and therefore by monotonicity $F(x(t) + z') \geq F(x^*) =
  OPT$. Moreover, since $F$ is up-concave and $z' \succeq 0$, we get
  that $\grad F(x(t)) \cdot z' \geq OPT - F(x(t))$. This completes the
  proof.
\end{proofsketch}

When $F$ is the multilinear relaxation of $f$, we can evaluate $F$ to
arbitrary precision by a polynomial number of random samples
\cite{vondrak_stoc08}.  Combining Lemmas \ref{lem:matroid_nogap} and
\ref{lem:maximize_upconcave}, we get the Theorem. Technical details
that compensate for the loss of approximation due to sampling are
ommitted.

\begin{theorem} (\cite{vondrak_stoc08}) There exists an algorithm for maximizing a
  nonnegative, monotone, submodular function $f:2^X \to \RR$ given by
  a value oracle, over a matroid $M$ given by an independence oracle,
  that achieves an approximation ratio of $1-1/e$ and runs in time
  polynomial in $n$ and $\log B$.
\end{theorem}

\subsection{New Result: Minimizing Nonnegative Symmetric Submodular
  Functions Subject to a Cardinality Constraint}\label{sec:minimizing_cardinality}
In this section, we consider the problem of minimizing a nonnegative
symmetric submodular function subject to a cardinality
constraint. First, we make the simple observation that, by
submodularity, the minimum of a symmetric submodular function is
always attained at $\emptyset$ and $X$. Therefore, as is usual when we
are working with symmetric submodular functions, we consider
minimization of $f$ over non-empty sets. Moreover, observe that, by
symmetry, an upperbound of $k$ on the cardinality is equivalent to a
lowerbound of $n-k$. Therefore, we assume without loss that are
minimizing $f$ over non-empty subsets of $X$ of cardinality at most
$k$.

Symmetric submodular functions often arise as cut-type functions. The
cut-function of an undirected graph is the canonical example. In this
context, our problem is equivalent to finding the minimum cut of the
graph that is sufficiently unbalanced: i.e. with smaller side having
cardinality at most $k$. We term this problem the
minimum-unbalanced-cut problem, and point out that it has obvious
implications for finding small ``communities'' in social networks.

A slight generalization of minimum unbalanced cut was studied in
\cite{kempe_unbalanced}. There, they consider the ``sourced'' version,
where a designated node $s$ is required to lie in the side of the cut
of interest (the side with at most $k$ nodes). They show that this
sourced-min-unbalanced-cut problem is NP-hard by reductions from
at-most-$k$-densest subgraph and max-clique. Moreover, they give an
algorithm achieving a bicriteria result parametrized by $\alpha >1$:
They find a cut of capacity at most $\alpha$ of the
optimal unbalanced cut, yet violating the cardinality constraint by a factor of
up to $\frac{\alpha}{\alpha -1}$. When $\alpha=2$, this gives a
$2$-approximation algorithm that overflows the constraint by a factor of
at most $2$. Their techniques do not directly yield a constant
approximation algorithm for the problem without violating the constraint.


\subsubsection{A 2-approximation algorithm}
In this section, we show a 2-approximation algorithm for minimizing a
nonnegative, symmetric submodular function subject to a cardinality
constraint. Without loss, we assume the constraint is an upper bound
of $k$ on the cardinality of the set. The algorithm operates in the
value query model, and runs in polynomial time. This result is
stronger than the result in \cite{kempe_unbalanced} in two ways: It
applies to general nonnegative symmetric submodular functions rather
than just graph cut functions, and it achieves a constant factor
approximation without violating the constraint. The reader may notice,
however, that this problem as-stated is not strictly more general than
the ``sourced'' problem considered in \cite{kempe_unbalanced}. We
leave open the question of whether a similar guarantee is possible for
the sourced problem.

\begin{algorithm}
\caption{2-approximation for minimizing nonnegative, symmetric, submodular f subject to
  cardinality constraint.}
\label{minsymm}
\begin{algorithmic}[1]
  \REQUIRE $f: 2^X \to \RR$ a nonnegative, symmetric, submodular
  function given by a value oracle. Integer $k$ such that $0 < k
  <n$.

  \ENSURE $Q$ minimizes $f$ over non-empty sets of size at most
  $k$

  \FORALL{$v_1 \in X$} 

  \STATE Find $x \in [0,1]^X$ minimizing Lov\'{a}sz extension $\L_f$
  subject to $x(v_1) = 1$ and ${ \vec{\bf 1}} \cdot x \leq k$.

  \STATE Construct the Lov\'{a}sz extension distribution $D^\L(x)$
  corresponding to point $x$.

  \IF{ There is $S$ in the support of $D^\L(x)$ with $|S| \leq k$ and
    $f(S) \leq 2 L_f(x) $} 
  \RETURN $S$
  \ELSE
  \STATE Find $S'$ in the support of $D^L(x)$  minimizing $f(S')$ subject to $|S'| \leq
  2k$ 

  \FORALL{$v_2 \in S' \sm \set{v_1}$} 
  \STATE Using submodular minimization, find $T$
  minimizing $f(T)$ subject to $v_1 \in T$ and $v_2 \notin T$.
  \IF{ $|T \intersect S'| \leq k$}
  \STATE $Q_{v_1,v_2} := T \intersect S'$
  \ELSE 
  \STATE $Q_{v_1,v_2} := \overline{T} \intersect S'$
  \ENDIF
  \ENDFOR
  \ENDIF
  \ENDFOR
  \STATE $Q := \argmin_{v_1,v_2} f(Q_{v_1,v_2})$
  \RETURN $Q$
\end{algorithmic}
\end{algorithm}

We will now argue that Algorithm \ref{minsymm} runs in polynomial
time. Step 2 can completed in polynomial time by standard convex optimization
techniques. For step 3, the polynomial-time construction in Section
\ref{sec:lovasz_ext} computes an explicit representation of
$D^\L(x)$. Moreover, from Section \ref{sec:lovasz_ext} we know
that $D^\L(x)$ has a support of size at most $n+1$, and thus steps 4
and 7 can
be completed in polynomial time. It is then easy to see that the
entire algorithm terminates in polynomial time.

Next, we argue correctness by nondeterministically stepping through
the algorithm. Let $S^*$ denote the optimal solution to the problem,
with $f(S^*)= OPT$. First, assume the algorithm guesses some $v_1
\in S^*$. Since $\L_f$ is an extension of $f$ and $S^*$ has cardinality
at most $k$, step 2 computes $x$ with $\L_f(x) \leq OPT$. Moreover, we
know from Section \ref{sec:lovasz_ext} that $\L_f(x)$ is the
expected value of $f$ over draws from $D^\L(x)$.

If $S$ with $|S| \leq k$ and $f(S) \leq 2 L_f(x)$ is found in step 4,
then we terminate correctly with a $2$-approximation. Otherwise, we
can show that step 7 finds $S'$ with $f(S') \leq OPT$.

\begin{lemma}\label{S'_good}
  Either there exists $S$ in the support of $D^\L(x)$ with $|S| \leq
  k$ and $f(S) \leq 2 L_f(x)$, or there exists $S'$ in the support of
  $D^\L(x)$ with $|S'| \leq 2 k$, and $f(S') \leq L_f(x)$. 
\end{lemma}
\begin{proof}
  Assume not. It is now easy to check that each set $R$ in the
  support of $D^\L(x)$ has \[f(R)>\left(2-\frac{|R|}{k}\right) \L_f(x)\] Taking
  expectations, we get that
  \[ \Ex_{R \sim D^\L(x)} f(R) > \left(2 - \Ex_{R \sim D^\L(x)}
  \frac{|R|}{k}\right) \L_f(x) \geq L_f(x)\]
The last inequality follows from the fact that the expected value of
$|R|$ is at most $k$, by definitoin of $D^\L(x)$. This is a
contradiction, since by definition the expectation of $f$ over draws
from $D^\L(x)$ is precisely $\L_f(x)$.
\end{proof}

Now, assuming no appropriate $S$ was found in step 4, we have $S'$ as
in the statement of Lemma \ref{S'_good} with $k< |S'| \leq 2k$ and
$v_1 \in S'$. Since $|S^*| \leq k$, we know that there exists $v_2 \in
S'$ such that $v_2 \notin S^*$. In particular, there exists a set
containing $v_1$ and not containing $v_2$ with value at most $OPT$.
Assume the algorithm guesses such a $v_2$. This immediately yields the
following Lemma.

\begin{lemma}\label{T_good}
  If $v_1 \in S^*$ and $v_2 \in S' \sm S^*$ then step 9 finds $T$ such
  that $f(T) \leq OPT$.
\end{lemma}

Therefore, combining Lemmas \ref{S'_good} and \ref{T_good}, we get the
following from submodularity and nonnegativity:
\[f( T \intersect S') \leq  f(T \intersect S') + f(T \union S') \leq f(T)
+ f(S') \leq OPT + OPT = 2OPT\]

Moreover, we know by symmetry of $f$ that $f(\overline{T}) = f(T) \leq
OPT$. Therefore, by the same calculation we get $f( \overline{T}
\intersect S') \leq 2OPT$. Now, observe that $T \intersect S'$ and
$\overline{T} \intersect S'$ partition $S'$ into non-trivial subsets by
definition of $T$.  This gives that the smaller of the two,
$Q_{v_1,v_2}$, has cardinality between $1$ and $k$, and moreover
$f(Q_{v_1,v_2}) \leq 2OPT$. The algorithm tries all $v_1$ and $v_2$,
so this immediately yields the Theorem.

\begin{theorem}
  Algorithm \ref{minsymm} is a polynomial-time $2$-approximation
  algorithm for minimizing a nonnegative, symmetric, submodular
  function subject to a cardinality constraint in the value oracle model.
\end{theorem}

\paragraph{Conclusion}
In this survey, we considered various continuous extensions of
submodular functions. We observed that those extensions yielding
algorithmic utility are often associated with natural, even
\emph{oblivious} distributions on the ground set. We presented a
unified treatment of two existing algorithmic results, one on minimization and
one on maximization, using this distributional lens. Moreover, we demonstrate the
power of this paradigm by obtaining a new result for constrained
minimization of submodular functions.

\paragraph{Acknowledgements}
The author would like to thank the qualifying exam committee:
Serge Plotkin, Tim Roughgarden, and Jan Vondr{\'a}k, for helpful
discussions. We single out Jan Vondr{\'a}k for his careful
guidance, and for pointing out an improvement in the approximation
ratio of the result in Section \ref{sec:minimizing_cardinality}.

{
\bibliography{bib}
\bibliographystyle{plain}
}

\end{document}